%% file: main.tex
\documentclass{amsart}
\usepackage{stmaryrd}
\usepackage{amsfonts}
\usepackage{bbm}
\usepackage{amscd}
\usepackage{mathrsfs}
\usepackage{latexsym,amssymb,amsmath,amscd,amscd,amsthm,amsxtra,xypic}
\usepackage[dvips]{graphicx}
\usepackage[utf8]{inputenc}
\usepackage[T1]{fontenc}
\usepackage{lmodern}
\usepackage{amssymb}
\usepackage[all]{xy}
\usepackage{nicefrac,mathtools,enumitem}
\usepackage{microtype}
\usepackage{subfiles}
\usepackage{color}
\usepackage{wasysym}

\calclayout

\newtheorem{thm}{Theorem}[section]
\newtheorem{lem}[thm]{Lemma}
\newtheorem{example}[thm]{Example}

\newtheorem{pro}[thm]{Proposition}

\newtheorem{rmk}[thm]{Remark}

\newcommand{\be }{\begin{equation}}
\newcommand{\ee }{\end{equation}}

\newcommand{\frkg}{\mathfrak g}

\newcommand{\br}[1]{   [ \cdot,    \cdot  ]   }

\newcommand{\g}{\mathfrak g}

\newcommand{\Aut}{\mathrm{Aut}}

\DeclareMathOperator{\tder}{tder}
\DeclareMathOperator{\der}{der}
\DeclareMathOperator{\sder}{sder}
\DeclareMathOperator{\TAut}{TAut}
\DeclareMathOperator{\SAut}{SAut}
\DeclareMathOperator{\Lie}{Lie}

\begin{document}

\title{Chern-Simons, Wess-Zumino and other cocycles from Kashiwara-Vergne and associators}


\author{Anton Alekseev}
\address{Section of
  Mathematics, University of Geneva, 2-4 rue du Li\`evre, c.p. 64, 1211 Gen\`eve 4, Switzerland}
\email{Anton.Alekseev@unige.ch}

\author{Florian Naef}
\address{Section of
  Mathematics, University of Geneva, 2-4 rue du Li\`evre, c.p. 64, 1211 Gen\`eve 4, Switzerland}
\email{Florian.Naef@unige.ch}

\author{Xiaomeng Xu}
\address{Department of mathematics, Massachusetts institute of technology, Cambridge, MA 02139, United States of America}
\email{xxu@mit.edu}

\author{Chenchang Zhu}
\address{Mathematics Institute,
Georg-August-University
G\"ottingen, Bunsenstrasse 3-5, 37073, G\"ottingen, Germany}
\email{czhu@gwdg.de}

\maketitle

\hskip 8cm
{\em To the memory of Ludwig Faddeev}

\begin{abstract}
Descent equations play an important role in the theory of characteristic classes and find applications in theoretical physics, {\em e.g} in the Chern-Simons field theory and in the theory of anomalies.  The second Chern class (the first Pontrjagin class) is defined as $p= \langle F, F\rangle$ where $F$ is the curvature 2-form and $\langle \cdot, \cdot\rangle$ is an invariant scalar product on the corresponding Lie algebra $\g$. The descent for $p$ gives rise to an element $\omega=\omega_3 + \omega_2 + \omega_1 + \omega_0$ of mixed degree. The 3-form part $\omega_3$ is the Chern-Simons form. The 2-form part $\omega_2$ is known as the Wess-Zumino action in physics. The 1-form component $\omega_1$ is related to the canonical central extension of the loop group $LG$. 

In this paper, we give a new interpretation of the low degree components $\omega_1$ and $\omega_0$. Our main tool is the universal differential calculus on free Lie algebras due to Kontsevich. We establish a correspondence between  solutions of the first Kashiwara-Vergne equation in Lie theory and universal solutions of the descent equation for the second Chern class $p$. In more detail, we define a 1-cocycle $C$ which maps automorphisms of the free Lie algebra to one forms. A solution of the Kashiwara-Vergne equation $F$ is mapped to $\omega_1=C(F)$. Furthermore, the component $\omega_0$ is related to the associator corresponding to $F$. It is surprising that while $F$ and $\Phi$ satisfy the highly non-linear twist and pentagon equations, the elements $\omega_1$ and $\omega_0$ solve the linear descent equation.
\end{abstract}

\tableofcontents

\subfile{sect0}

\subsection*{Acknowledgement}
The authors would like to thank Jürg Fröhlich, Krzysztof Gawedzki, Samson Shatashvili, and Pavol Ševera for helpful discussions and suggestions.
A.A. and F.N. were supported by the grant MODFLAT of the European Research Council and the NCCR SwissMAP of the Swiss National Science Foundation (SNSF). A.A. was supported by the grants number 165666 and 159581 of the SNSF. X.X. was supported by the Early Postdoc.Mobility grant of SNSF. C.Z. was supported by the Deutsche Forschungsgemeinschaft (DFG) through the Institutional Strategy of the University of G\"ottingen and DFG Individual Grant (ZH 274/1-1) “Homotopy Lie Theory”.


\subfile{sect1_anton}

\subfile{sect2}


\end{document}

%% file: sect0.tex
\section{Motivation: descent equations}
\label{sect:intro}

Le $G$ be a connected Lie group with Lie algebra $\frkg$, and let $P \to M$ be a principal $G$-bundle with connection $A \in \Omega^1(P, \frkg)$. Without loss of generality, one can assume that $G$ is a matrix Lie group (by Ado Theorem, $G$ always admits a faithful representation). Then, gauge transformations can be written in the form
\begin{equation} \label{gauge}
A \mapsto A^g = g^{-1} A g + g^{-1} dg,
\end{equation}
and the curvature of $A$ is defined by
$$
F=dA + \frac{1}{2} [A, A] = dA + A^2.
$$

Polynomials of $A$ and $F$ form the Weil algebra $W\frkg$ (for details see {\em e.g.} \cite{Guillemin-Sternberg}). The defining equation of the curvature $F$ and the Bianchi identity give rise to  the definition of the Weil differential:
$$
d_W A = F - \frac{1}{2} [A, A], \hskip 0.3cm
d_W F = - [A, F].
$$

Elements of $W\frkg$ basic under the $G$-action (that is, $G$-invariant and horizontal) give rise to differential forms on the total space of the bundle $P$ which descend to the base $M$. Moreover, it turns out that such forms are automatically closed and gauge invariant. Their cohomology classes in $H^\bullet (M, \mathbb{R})$ are characteristic classes of the bundle $P$.

Assume that the Lie algebra $\frkg$ carries an invariant scalar product $\langle \cdot, \cdot \rangle$. Then, the element of $W\frkg$ 
$$
p = \langle F, F \rangle
$$
is basic. Its cohomology class is the second Chern class (or the first Pontrjagin class) of $P$ if we choose a rescaled inner product with a suitable coefficient. In $W\frkg$, the element $p$ admits a primitive:
$$
p = d \, {\rm CS},
$$
where ${\rm CS} \in W\frkg$ is the Chern-Simons form
\begin{equation} \label{cs}
{\rm CS} = \langle A, dA \rangle + \frac{1}{3} \langle A, [A, A] \rangle.
\end{equation}
Note that $p$ does not admit a basic primitive within $W\frkg$. This is why the corresponding characteristic class is non vanishing in general. However, if one removes the requirement of being basic, the Weil algebra is acyclic and all closed elements of non vanishing degree admit primitives.

The Chern-Simons form \eqref{cs} plays a major role in applications to Quantum Topology and  Quantum Field Theory (QFT). In this paper, we will study its version in the complex
\begin{equation} \label{C^n}
\mathcal{C}^n = \oplus_{i+j+k=n} \, W^i\frkg \otimes \Omega^j(G^k).
\end{equation}
Here $W^i\frkg$ is the subspace of $W\frkg$ spanned by the elements of degree $i$, and $G^k$ is a direct product of $k$ copies of $G$. The graded vector space $C = \oplus _n C^n$ admits a differential
$$
D = d_W+ d_{dR} + \Delta,
$$
where $d_W$ is the Weil differential on $W\frkg$, $d_{dR}$ is the de Rham differential on $\Omega(G^k)$, $\operatorname{deg}_\text{dR}$ the de Rham degree $i+j$, and $\Delta$ is the group cohomology differential given by
$$
(\Delta \alpha)(A, g_1, \dots, g_{k+1}) =
\alpha(A^{g_1}, g_2, \dots, g_{k+1}) - 
\alpha(A, g_1g_2, \dots, g_{k+1}) + \dots +
(-1)^{k+1} \alpha(A, g_1, \dots, g_k).
$$
Here $\alpha \in W\frkg \otimes \Omega(G^k)$, and the notation $\alpha(A, g_1g_2, \dots, g_{k+1})$ stands for the pull-back of $\alpha$ under the co-face map $G^{k+1} \to G^k$ defined by 
$$
(g_1, g_2, g_3, \dots, g_{k+1}) \mapsto (g_1g_2, g_3, \dots, g_{k+1}).
$$

In the complex $C$, we can ask again for the primitive of the element $p$. That is, for a solution of the equation $p=D\omega$. Solutions of this equation serve to define the Chern-Simons functional on simplicial spaces, and also in the theory of anomalies in QFT (see \cite{Faddeev-Shatashvili}, \cite{Faddeev}, \cite{Zumino}, \cite{Jackiw}).

If we denote $d=d_W + d_{dR}$, we obtain a system of equations known as descent equations:
\begin{equation} \label{eq:descend-omega}
\begin{array}{rll}
d \omega_3 & = & p, \\
-\Delta \omega_3 + d \omega_2 & = & 0, \\
\Delta \omega_2 + d \omega_1 & = & 0, \\
-\Delta \omega_1 + d \omega_0 & = & 0, \\
\Delta \omega_0 & = & 0.
\end{array}
\end{equation}
Here $\omega_i$ is an element of  $W\frkg \otimes \Omega(G^{3-i})$ of degree $i$ (the sum of the Weil degree in $W\frkg$ and the de Rham degree in $\Omega(G^{3-i})$. In particular, $\omega_3 \in W\frkg$ and the first equation reads
$$
d_W \omega_3 = p.
$$
Hence, we conclude that $\omega_3 = {\rm CS}$. In order to understand the next equation, we write
$$
\Delta \omega_3 = {\rm CS}(A^g) - {\rm CS}(A) =
d \langle A, g^{-1}dg \rangle - \frac{1}{6} \langle g^{-1}dg, [g^{-1} dg, g^{-1}dg] \rangle.
$$
Finding the primitive of this expression under the differential $d$ depends on vanishing of the cohomology class of the Cartan 3-form on $G$:
$$
\eta= \frac{1}{6} \langle g^{-1}dg, [g^{-1} dg, g^{-1}dg] \rangle.
$$
The cohomology class $[\eta]$ is non vanishing in general. In particular, it is non vanishing on all semisimple compact Lie groups. Hence, in this case descent equations admit no solution. There are several ways to address this difficulty. For QFT applications, one should replace differential forms on $G$ by differential characters, and consider exponentials of periods of these forms. 


In a geometric setting, one realises $p$ on a specific $G$-principal bundle $P$ over a manifold $X$. The geometric data corresponding to $p$ is the Chern-Simon 2-gerbe. When $p(P)$ vanishes, one is expected to find trivialisations of the Chern-Simon 2-gerbe, which are known as String structures over $P$ \cite{Stolz-Teichner}. As shown in \cite{Sheng-Xu-Zhu}, descent equations \eqref{eq:descend-omega} govern the connection data of String principal 2-bundles.  

In this paper, we choose another approach: we replace the group $G$ by the corresponding formal group $G_{\rm formal}$. Since formal manifolds are modeled on one chart (which is itself a formal vector space), the cohomology of $G_{\rm formal}$ is trivial and the cohomology class of the corresponding Cartan 3-form vanishes. In more detail, the exponential map
$$
\exp: \frkg_{\rm formal} \to G_{\rm formal}
$$
establishes an isomorphism between $G_{\rm formal}$ and its Lie algebra. The primitive of the Cartan 3-form can be computed using the Poincar\'e homotopy operator $h_P$:
$$
\eta = d h_P(\eta).
$$

In order to make this approach more precise, we  will use a version of the Kontsevich universal differential calculus. This leads to an (almost) unique solution for the component $\omega_2$ of the extended Chern-Simons form which is called the Wess-Zumino action in the physics literature:
$$
{\rm WZ}(A,g) = \langle A, dgg^{-1} \rangle -\frac{1}{6} \, h_P(\langle g^{-1}dg, [g^{-1} dg, g^{-1}dg]\rangle).
$$

The goal of this paper is to give an interpretation of the components $\omega_1$  and $\omega_0$ of the extended Chern-Simons form. In more detail, we will define a 1-cocycle $C$ mapping the group of tangential automorphisms of the free Lie algebra with $n$ generators ${\rm TAut}_n$ (see the next sections for a precise definition) to the space of universal 1-forms $\Omega^1(G_{\rm formal}^n)$. We will then establish the following surprising relation between certain special elements in ${\rm TAut}_2$ and ${\rm TAut}_3$ and universal solutions of descent equations.

It turns out that the equation $\delta \omega_2 + d \omega_1 = 0$, with $\omega_2$ given by the Wess-Zumino action, admits solutions of the form $\omega_1=C(g)$, where $g\in {\rm TAut}_2$ with the property
$$
g(x_1 + x_2)=\log(e^{x_1}e^{x_2}).
$$
Here the right hand side is the Baker-Campbell-Hausdorff series. Furthermore, the next descent equation $\Delta \omega_1 + d \omega_0 =0$ translates into the twist equation from the theory of quasi-Hopf algebras:
$$
g^{1,2}g^{12,3}=g^{2,3}g^{1,23} \Phi^{1,2,3},
$$
where $g^{1,2}$, $g^{12,3}$, etc. are images of $g$ under various co-face maps, and $\Phi \in {\rm TAut}_3$ is related to $\omega_0$ via $\omega_0=(g^{2,3} g^{1,23}).h_P (C(\Phi))$. Last but not least, the descent equation $\Delta \omega_0=0$ translates into the pentagon equation for $\Phi$:
$$
\Phi^{12,3,4}\Phi^{1,2,34} = \Phi^{2,3,4} \Phi^{1,23,4} \Phi^{1,2,3}.
$$

The structure of the paper is as follows:
the second chapter contains a recollection of non-commutative differential calculus and the construction of the universal Bott-Shulman double complex (\cite{Bott}, \cite{BSS}, \cite{BS})  in the abelian and non-abelian cases. In particular, the relevant cohomology groups are computed and their relation with the classical Bott-Shulman complex is discussed.
The third chapter contains a short discussion of the Kashiwara-Vergne theory and the construction of the 1-cocycle $C$ which is the main new object in this note. After establishing some of its properties, we prove the formulas for $\omega_1$ and $\omega_0$.

%% file: sect1_anton.tex
\section{Universal calculus for gauge theory}

\subsection{Non-commutative differential calculus}

Let $\Bbbk$ be a field of characteristic zero, and let ${\rm Lie}(x_1, \dots, x_n)$ be a free graded Lie superalgebra with generators $x_1, \dots, x_n$. We assume that the degrees $|x_i|$ are non-negative for all $i$. If all the degrees vanish, we use the shorthand notation ${\rm Lie}_n$ instead of ${\rm Lie}(x_1, x_2, \dots, x_n)$. One considers ${\rm Lie}(x_1, x_2, \dots, x_n)$ as a coordinate algebra of a non-commutative space. Note that this free Lie algebra carries two gradings: the first one is induced by the degrees of generators, the other one is obtained by counting the number of letters in a Lie word.

Following Drinfeld and Kontsevich, we define functions on this non-commutative space in the following way:
$$
\begin{array}{lll}
\mathcal{F}\langle x_1, \dots, x_n\rangle &= & {\rm Lie}(x_1, \dots, x_n) \otimes {\rm Lie}(x_1, \dots, x_n)/ \\
& & {\rm Span}\{ a \otimes b - (-1)^{|a||b|} b\otimes a, \,\, a\otimes [b,c] - [a,b] \otimes c \} .
\end{array}
$$
The natural projection is denoted by $\langle \cdot, \cdot \rangle$.
Again, we use the notation $\mathcal{F}_n$ when all generators have vanishing degree.

\begin{example}
The space $\mathcal{F}_1$ is a line with generator $\langle x_1, x_1 \rangle$.
The space $\mathcal{F}_2$ is infinite dimensional. In low degrees, we have the following elements which contain two letters: $\langle x_1, x_1\rangle$, $\langle x_2, x_2\rangle$ and $\langle x_1, x_2\rangle$. The space $\mathcal{F}_3$ contains a unique line of multilinear elements spanned by $\langle x_1, [x_2, x_3] \rangle$. 

\end{example}

The space of differential forms is defined as functions on the shifted tangent bundle. More concretely,
$$
\Omega\langle x_1, \dots, x_n\rangle = \mathcal{F}\langle x_1, \dots, x_n, dx_1, \dots, dx_n\rangle,
$$
where we have $|dx_i| = |x_i| +1$. Equivalently, this is the space of functions on the differential envelope of ${\rm Lie}(x_1, \cdots, x_n)$ which is ${\rm Lie}(x_1, \dots, x_n, dx_1, \dots, dx_n)$ with obvious differential. The differential descends to $\Omega\langle x_1, \dots, x_n\rangle$ making it into a complex of vector spaces.

\begin{example}
Let $A$ be a generator of degree $1$. The corresponding space of differential forms $\Omega\langle A\rangle = \mathcal{F}\langle A, dA \rangle$ will be referred to as the universal Weil algebra. Sometimes it is more convenient to use the following generator of degree 2
$$
F := dA + \frac{1}{2} [A,A]
$$
instead of $dA$. Its differential is given by $dF=-[A, F]$. The low degree elements of $\Omega\langle A\rangle$ are as follows. In degree 2: $\langle A, A\rangle$, in degree 3: $\langle A, F\rangle$ and $\langle A, [A,A]\rangle$, and in degree 4: $\langle F, F \rangle$ and $\langle F, [A, A] \rangle$.
\end{example}

\begin{lem}\label{lem:d-acyclic}
The complex $(\Omega\langle x_1, \dots, x_k\rangle, d)$ is acyclic for $k\ge 1$.
\end{lem}

\begin{proof}
Consider the derivation $e$ of ${\rm Lie}(x_1, \dots, x_k, dx_1, \dots, dx_k)$ defined on generators as $e(dx_i)=x_i, e(x_i)=0$. A direct calculation shows that $de+ed =n$, where $n$ is the Euler derivation counting the number of generators in a Lie word. The derivations $e$ and $n$ descend to operators on $\Omega\langle x_1, \dots, x_k\rangle$. If $\alpha \in \Omega\langle x_1, \dots, x_k \rangle$ is a cocycle, we have $d(e(\alpha)) = n\alpha$. Since all elements of $\Omega\langle x_1, \dots, x_k \rangle$ contain at least two generators, $n\neq 0$ and we obtain an explicit homotopy which provides primitives of all cocycles.

\end{proof}

\subsection{Universal abelian gauge transformations}
In this section, we define a universal version of abelian connections and gauge transformations.
Let $A$ denote a generator of degree 1, and $x_1, \dots, x_n$ generators of degree 0. We define a co-simplicial complex of the following form:

\begin{equation} \label{eq:lie-cosim}
    \xymatrix{\Omega\langle A \rangle \,\, \ar@2[r]^(.4){\delta_0, \delta_1} & \,\, \Omega\langle A, x_1 \rangle \,\, 
      \ar@3[r]^(.4){\delta_0, \delta_1, \delta_2} &  \,\, 
      \Omega\langle A, x_1, x_2 \rangle  \dots},
\end{equation}
where the co-face maps are defined as follows:
\begin{equation} \label{eq:cofaces-lie-linear}
 \begin{array}{lllll}
\delta_0: &A \mapsto A+dx_1, &  x_1 \mapsto x_2 & \dots & x_n \mapsto x_{n+1},  \\
\delta_1: &A \mapsto A,  & x_1 \mapsto x_1 + x_2& \dots & x_n \mapsto x_{n+1}, 
\\
&\vdots & & & \\
\delta_n: & A\mapsto A, &  x_1 \mapsto x_1 & \dots & x_n \mapsto x_n + x_{n+1}, \\
\delta_{n+1}: & A\mapsto A, &  x_1 \mapsto x_1 & \dots & x_n \mapsto x_n. 
\end{array}
\end{equation}
The requirement that co-face maps commute with the differential determines them on  $dA, dx_1, \dots, dx_n$.
In particular, we have $\delta_i (dA)=dA $ for all $i$. The co-simplicial differential is defined by
$$
\delta = \sum_i (-1)^i \delta_i = \delta_0 - \delta_1 + \delta_2 - \dots 
$$
The differentials $d$ and $\delta$ commute hence we get a double complex. Let $D=d + \delta$ be the differential on its total complex:
\begin{equation}  \label{eq:total_simplicial}
\Omega^\bullet\langle A, x_\bullet \rangle =\oplus_{n=0}^\infty \, \Omega^\bullet\langle A, x_1, \dots, x_n\rangle.
\end{equation}
Note that in $D = d + \delta$ there is a sign suppressed in the second term that depends on the de Rham degree. It will be clear from the context which sign applies.
\begin{rmk} A more geometric construction of the complex \eqref{eq:lie-cosim} can be obtained as follows. Consider the opposite category of (differential) graded Lie algebras as a category of (dg) spaces. Let $G = \Lie(x)$ be the free Lie algebra in one generator considered as an object in this category. There are two non-isomorphic group structures on $G$: the one induced by linear function $x_1 + x_2$ and the one induced by the Baker-Campbell-Hausdorff Lie series $\log(e^{x_1} e^{x_2})$. The first one  turns $G$ into an abelian group. 

Let $T[1](-)$ denote the functor of taking universal differential envelopes. The space of left $G$-orbits in $T[1]G$ denoted by $V := G\backslash T[1]G$ inherits a right action of $T[1]G$, which we call the gauge action. One identifies $V = \Lie(A)$ for a generator $A$ of degree 1, and the action is given by
$$ 
A \mapsto A + dx
$$
for the first group structure on $G$, and
$$
A \mapsto e^{-x}A e^x + e^{-x}d(e^x)
$$
for the second one. This action lifts to $W=T[1]V$ as a right dg action. Functions on $W$ form the universal Weil algebra.

Note that in the non-abelain case $V$ carries a canonical Chevalley-Eilenberg differential  $d_\text{CE}: A \mapsto -\tfrac{1}{2}[A,A]$. Hence, $W=T[1]V$ carries the de Rham differential $d$ and the twisted differential  $d+ d_\text{CE}$ (which corresponds to the standard Weil differential). In this way, one gets the two geometrically equivalent coordinate systems on $T[1]V$: the first one is  $\{ A,dA\}$ and the second one is $\{ A, F = dA + \tfrac{1}{2}[A,A]\}$.

The space $* = \Lie(\varnothing) = 0$ carries a trivial left $T[1]G$ action. The complex \eqref{eq:total_simplicial} is then given by functions on the two-sided bar construction of $T[1]G$ acting on $T[1]V$ and on $*$:
$$
\mathcal{F}(B( T[1]V, T[1]G, *)).
$$
 The result is a simplicial dg-complex. Note that this description is completely categorical. In particular, if one replaces the category opposite to dg Lie algebras (which we were using above) by the category of dg manifolds, we recover the classical Bott-Shulman complex. 

In the abelian case, the construction can be simplified:
\begin{align*}
    B( T[1]W, T[1]G, *) &= B( \Lie(A, dA) , \Lie(x, dx), *) \\
    &= B( \Lie(A), \Lie(dx), *) \times \Lie(dA) \times B(*, \Lie(x), *),
\end{align*}
where the first factor is contractible and the last factor serves as a counterpart of $BG$. This analogy will be made more precise below.
\end{rmk}

In what follows we present several calculations for the complex $\Omega^\bullet\langle A, x_\bullet\rangle$ and study its properties.

\begin{lem}
The complex $(\Omega^\bullet\langle A, x_\bullet\rangle, D =d+ \delta)$ is acyclic.
\end{lem}

\begin{proof}
Recall that the complex $(\Omega^\bullet\langle A, x_1, \dots, x_n\rangle, d)$ is acyclic for all $n\ge 0$. We can compute the cohomology of the complex $(\Omega^\bullet\langle A, x_\bullet\rangle, D =d+ \delta)$ by using the spectral sequence with the first page $H^{\ge 0}(\Omega^\bullet\langle A, x_1, \dots, x_n\rangle, d) =0$. Since the first page vanishes, the cohomology of the total complex vanishes as well.

\end{proof}

The Lemma above makes use of the acyclicity of the universal de Rham complex (which we will consider as columns of the double complex). In what follows, we will also need information about the row direction given by the cosimplicial differential $\delta$. 

%
%
%

\begin{lem}\label{lem:retraction}
The injective chain map 
$
(\mathcal{F}\langle dA, x_\bullet\rangle, \delta)  \to (\Omega\langle A, x_\bullet\rangle, \delta)
$
induces an isomorphism in cohomology.
\end{lem}

\begin{proof}
Note that our complex is the diagonal part of the bi-cosimplicial complex
$$
C_{m,n}=\mathcal{F}\langle A, dx_1, \dots, dx_m, dA, x_1, \dots, x_n\rangle,
$$
where the coface maps of the first cosimplicial component act on generators $A, dx_1, \dots, dx_m$ and of the second cosimplicial component on generators $dA, x_1, \dots, x_n$. By the Eilenberg-Zilber Theorem, the cosimplicial cohomology of the diagonal $\oplus_n C_{n,n}$ is isomorphic to the bi-cosimplicial cohomology of the total complex $\oplus_{m,n} C_{m,n}$ with differential $\delta=\delta' + \delta''$. Here $\delta'$ acts on generators $A, dx_1, \dots, dx_m$ and $\delta''$ acts on generators $dA, x_1, \dots, x_n$. The following operator 
$$
(h \alpha)(A, dx_1, \dots, dx_{m-1}, dA, x_1, \dots, x_n) = \alpha(0, A, dx_1, \dots, dx_{m-1}, dA, x_1, \dots, x_n)
$$
provides a homotopy between the identity and the projection to constant functions of $A$ in $C_{0, \bullet}$ for the differential $ \delta'$.
Hence, it defines a deformation retraction to the subcomplex $\mathcal{F}\langle dA, x_\bullet\rangle$ and the injection of $\mathcal{F}\langle dA, x_\bullet\rangle$ in $\Omega\langle A, x_\bullet\rangle$ induces an isomorphism in cohomology, as required.
\end{proof}

 Let $\mathcal{H}\langle dA, x_1, \dots, x_n\rangle \subset \mathcal{F}\langle dA, x_1 \dots, x_n \rangle$ be the subspace spanned by the elements linear with respect to $x_1, \dots, x_n$ and completely skew-symmetric under the action of the permutation group $S_n$.

\begin{example}
Here are some examples of elements in $\mathcal{H}\langle dA, x_\bullet\rangle$: $\langle dA, [x_1, x_2] \rangle$ and $\varphi=\langle x_1, [x_2, x_3]\rangle$. The class $\varphi$ plays an important role the following section.
\end{example}

\begin{lem}
Elements $\alpha \in  \mathcal{H}\langle dA, x_\bullet\rangle$ are $\delta$-closed.
\end{lem}

\begin{proof}
We give an example of a calculation for $\mathcal{H}\langle dA, x_1 \rangle$:
$$
\begin{array}{lll}
(\delta \alpha)(dA, x_1, x_2) & = & \alpha(dA, x_1) - \alpha(dA, x_1 + x_2) + \alpha(dA, x_2) \\
& = & \alpha(dA, x_1) - \alpha(dA, x_1) - \alpha(dA, x_2) + \alpha(dA, x_2) \\
& = & 0.
\end{array}
$$
Here we have used the linearity of $\alpha$ with respect to the argument $x$. The calculation works in the same way in higher degrees.
\end{proof}

\begin{lem}  \label{lem:standard}
The injective chain map $(\mathcal{H}\langle dA, x_1, \dots, x_n\rangle, 0) \subset (\mathcal{F}\langle dA, x_1 \dots, x_n \rangle, \delta)$ induces an isomorphism in cohomology.
\end{lem}

\begin{proof}
Standard (see \cite{Severa-Willwacher(2011)}, \cite{barnatan}, \cite{vergne}, \cite{drinfeld}).
\end{proof}

The lemma above implies that $\delta$-closed elements in $\Omega\langle A, dA\rangle$ must be functions of $dA$. Such a function is unique up to a multiple, and it is given by the abelian second Chern class:
$$
p=\langle dA, dA\rangle.
$$
Since the total double complex  $(\Omega\langle dA, x_\bullet\rangle, D)$ is acyclic, one can ask for a primitive (the cochain of transgression) of the function $p$ which is given by the following formula
$$
\omega= \langle A, dA\rangle + \langle A, dx_1\rangle - \langle x_1, dx_2 \rangle.
$$
Since the differentials $d$ and $\delta$ both preserve the number of letters, the primitive can be chosen in the same graded component as the class $\langle dA, dA\rangle$. The form $\langle A, dA\rangle$ is the abelian Chern-Simons element, and $\langle A, dx_1\rangle$ is the abelian Wess-Zumino action, the expression $\langle x_1, dx_2 \rangle$ stands for the Kac-Peterson cocycle on the current algebra. The primitive $\omega$ is unique up to  exact terms
$$
\omega' = \omega +  D (a_1 \langle A, x_1 \rangle + a_2 \langle x_1, x_2 \rangle + a_3 \langle x_1, x_1\rangle + a_4 \langle x_2, x_2 \rangle),
$$
where $a_i$'s are arbitrary coefficients.

\subsection{Non-abelian descent equations}

In this section, we consider a more complicated cosimplicial structure on $\Omega\langle A, x_\bullet\rangle$ which captures the features of non-abelian gauge theory. In more detail, we define new coface maps:
\begin{equation} \label{eq:cofaces-lie-nonlinear}
 \begin{array}{lllll}
\Delta_0: & A \mapsto e^{-x_1} A e^{x_1} + e^{-x_1} d e^{x_1}, & x_1 \mapsto x_2 & \dots & x_n \mapsto x_{n+1},  \\
\Delta_1: & A \mapsto A,  & x_1 \mapsto \log(e^{x_1}e^{x_2}) & \dots & x_n \mapsto x_{n+1}  
\\
&\vdots & & & \\
\Delta_n: & A\mapsto A, &   x_1 \mapsto x_1 & \dots & x_n \mapsto \log(e^{x_n}e^{x_{n+1}}), \\
\Delta_{n+1}: & A\mapsto A, &   x_1 \mapsto x_1 & \dots & x_n \mapsto x_n. 
\end{array}
\end{equation}
Here the formulas
\begin{equation}\label{eq:adj-action}
\begin{split}
 e^{-x_1} A e^{x_1} &= {\rm exp}(-{\rm ad}_{x_1} ) A = A - [x_1, A]+ \frac{1}{2} [x_1, [x_1, A]] - \dots, \\
 e^{-x_1} d e^{x_1} &= f({\rm ad}_{x_1}) dx_1, \quad \text{for}\;f(z)= \frac{1-e^{-z}}{z}
\end{split}
\end{equation}
define the non-abelian gauge action and
$$
\log(e^{x_1}e^{x_2}) = x_1 + x_2 + \frac{1}{2} [x_1, x_2] + \dots 
$$
is the Baker-Campbell-Hausdorff series. The cosimplicial differential is defined as the alternated sum of coface maps,
$$
\Delta= \sum_i (-1)^i \Delta_i.
$$
As before, the de Rham and cosimplicial differentials commute and let again $D=d + \Delta$ denote the total differential (with suppressed de Rham sign).
The complex $(\Omega^\bullet\langle A, x_\bullet\rangle, D =d+ \Delta)$ is again acyclic since it is acyclic under the de Rham differential $d$. 

Recall that in the non-abelian framework it is convenient to use the generator
$$
F=dA + \frac{1}{2} [A,A]
$$
instead of $dA$ since it has a nice transformation law under gauge transformations: $F \mapsto e^{-x_1} F e^{x_1}$. In what follows we will be interested in cohomology of the cosimiplical differential $\Delta$. It is convenient to introduce a decreasing filtration of the complex $\Omega\langle A, x_\bullet \rangle$ by the number of generators in the given expression (the elements of filtration degree $k$ contain at least $k$ generators). It is clear that the associated graded complex coincides with the cosimplicial complex for abelian gauge transformations.

\begin{lem}
Let $\alpha \in \Omega\langle A \rangle $ be a $\Delta$-cocycle. Then,  $\alpha=\langle F\rangle = \lambda \langle F, F\rangle$ for $\lambda \in \Bbbk$.  
\end{lem}

\begin{proof}
Assume that $\Delta \alpha=0$. Then its principal part $\alpha_{\rm low}$ (containing the lowest number of generators) is a $\delta$-cocycle. Hence, $\alpha_{\rm low}=\lambda \langle dA, dA \rangle$ for some $\lambda \in k$. Note that $\alpha' = \lambda \langle F, F \rangle$ is a $\Delta$-cocycle. By the argument above, the lowest degree part of $\alpha - \alpha'$ vanishes. Hence, $\alpha=\alpha' = \lambda \langle F, F\rangle$, as required.
\end{proof}

\begin{lem}\label{lem:cohgen}
There is an element
$\phi=\langle x_1, [x_2, x_3]\rangle + \dots \in \mathcal{F}\langle x_1, x_2, x_3 \rangle$ such that 
$\Delta \phi=0$. Its cohomology class is the generator of the cohomology group $H(\mathcal{F}\langle x_\bullet \rangle, \Delta) \cong k[\phi]$.
\end{lem}

\begin{proof}
Consider the decreasing filtration on $\mathcal{F}\langle x_\bullet\rangle$ defined by the number of generators in the expression. Recall that the associated graded complex of $(\mathcal{F}\langle x_\bullet \rangle, \Delta)$ is $(\mathcal{F}\langle x_\bullet \rangle, \delta)$. By Lemma \ref{lem:standard}, the cohomology of the latter complex is spanned by the class of $\varphi=\langle x_1, [x_2, x_3]\rangle$. That is, $[\varphi]$ is the only class on the first page of the spectral sequence defined by the filtration. Hence, it cannot be killed at any later page and it lifts to a cohomology class  which spans $H(\mathcal{F}\langle x_\bullet \rangle, \Delta)$.
\end{proof}

Recall that the element $\langle F, F\rangle$ is also a cocycle under the de Rham differential $d$. Hence, $DF = dF + \Delta F =0$. Since the 
complex $(\Omega\langle A, x_\bullet \rangle, D)$ is acyclic, there is a primitive $\omega=\omega_3+\omega_2+\omega_1+\omega_0$ such that $D\omega = \langle F, F\rangle=0$, with $\omega_3\in \Omega^3\langle A \rangle$, $\omega_2 \in \Omega^2\langle A, x_1 \rangle$, $\omega_1 \in \Omega^1\langle A, x_1, x_2 \rangle$, $\omega_0 \in \Omega^0\langle A, x_1, x_2, x_3\rangle$. The following diagram visualizes the lower left coner of the double complex that we are using:

\begin{equation}\label{db-cx:uni-BG}
 \xymatrix{
 &&&&&&
\\
\Omega^3\langle A\rangle \ar[r]^{\Delta}\ar[u]^{d}
 &\Omega^3\langle A, x_1\rangle \ar[r]^{\Delta}\ar[u]^{d}
 &\dots
&
&
&
 \\
  \Omega^2\langle A\rangle \ar[r]^{\Delta}\ar[u]^{d}
  &\Omega^2\langle A, x_1\rangle \ar[r]^{\Delta}\ar[u]^{d}
 &\Omega^2\langle A, x_1, x_2\rangle \ar[r]^{\Delta}\ar[u]^{d}
&\dots 
& 
&
 \\
 \Omega^1\langle A\rangle \ar[r]^{\Delta}\ar[u]^{d}
 & \Omega^1\langle A, x_1\rangle  \ar[r]^{\Delta}\ar[u]^{d}
 &\Omega^1\langle A, x_1, x_2\rangle \ar[r]^{\Delta}\ar[u]^{d}
&\Omega^1\langle A, x_1, x_2, x_3\rangle \ar[r]^-{\Delta}\ar[u]^{d}
&\dots 
&
 \\
 \Omega^0\langle A\rangle \ar[r]^{\Delta}\ar[u]^{d}
 &\Omega^0\langle A, x_1\rangle  \ar[r]^{\Delta} \ar[u]^{d}
&\Omega^0\langle A, x_1, x_2\rangle \ar[r]^{\Delta}\ar[u]^{d}
&\Omega^0\langle A, x_1, x_2, x_3\rangle \ar[r]^-{\Delta} \ar[u]^{d}
&\dots 
&
}
 \end{equation}
which simplifies by degree reasons to
\begin{equation}\label{db-cx:uni-BG2}
 \xymatrix{
 &&&&&&
\\
\Omega^3\langle A\rangle \ar[r]^{\Delta}\ar[u]^{d}
 &\Omega^3\langle A, x_1\rangle \ar[r]^{\Delta}\ar[u]^{d}
 &\dots
&
&
&
 \\
  0 \ar[r]^{\Delta}\ar[u]^{d}
  &\Omega^2\langle A, x_1\rangle \ar[r]^{\Delta}\ar[u]^{d}
 &\Omega^2\langle A, x_1, x_2\rangle \ar[r]^{\Delta}\ar[u]^{d}
&\dots 
& 
&
 \\
 0 \ar[r]^-{\Delta}\ar[u]^{d}
 & \Bbbk \langle dx_1,x_1 \rangle \oplus \Bbbk \langle A, x_1 \rangle  \ar[r]^{\Delta}\ar[u]^{d}
 &\Omega^1\langle A, x_1, x_2\rangle \ar[r]^{\Delta}\ar[u]^{d}
&\Omega^1\langle A, x_1, x_2, x_3\rangle \ar[r]^-{\Delta}\ar[u]^{d}
&\dots 
&
 \\
  0\ar[r]^{\Delta}\ar[u]^{d}
 & \Bbbk \langle x_1,x_1 \rangle  \ar[r]^{\Delta} \ar[u]^{d}
&\mathcal{F}\langle x_1, x_2\rangle \ar[r]^{\Delta}\ar[u]^{d}
&\mathcal{F}\langle x_1, x_2, x_3\rangle \ar[r]^-{\Delta} \ar[u]^{d}
&\dots 
&
}
 \end{equation}
 
The main result of this section is the following theorem: 

\begin{thm}\label{thm:ascent}
Let $\omega_1 \in \Omega^1\langle A, x_1, x_2 \rangle$ such that $d\Delta \omega_1 =0$. Then, there exists a unique element $\omega=\omega_0+\omega_1+\omega_2+\omega_3$  such that $D\omega=\lambda \langle F, F \rangle$ for some  $\lambda \in {\Bbbk}$. Moreover, there exists an element $\omega_1$ which yields $\lambda=1$.  
\end{thm}

To prove this result, we need the following lemma:

\begin{lem}\label{lem:delta-exact}
The complex $(\Omega^k\langle A, x_\bullet\rangle, \Delta)$ is exact 
for $k$ odd and for $k=2$.
\end{lem}

\begin{proof}
Recall that the associated graded of the complex $(\Omega^k\langle A, x_\bullet \rangle, \Delta)$ with respect to filtration defined by the number of generators yields the complex $(\Omega^k\langle A, x_\bullet \rangle, \delta)$. By Lemma \ref{lem:retraction}, this latter complex is acyclic for $k$ odd. Hence, so is the complex $(\Omega^k\langle A, x_\bullet \rangle, \Delta)$. Among other things, this implies that the 3d row is exact in $\Omega^3\langle A, x_1\rangle$.


For $k=2$, the associated graded complex has non-trivial cohomology. The non-trivial cohomology classes are represented by $\langle dA, x_1 \rangle \in \Omega^2\langle A, x_1\rangle$ and $\langle dA, [x_1, x_2] \rangle \in \Omega^2\langle A, x_1, x_2\rangle$.  Choose the lift $\langle F, x_1 \rangle$ of the class of $\langle dA, x_1 \rangle$ (they only differ by higher degree terms) and compute:


\[
\begin{split}
 \Delta \langle F, x \rangle = &\langle e^{-x_1} F e^{x_1}, x_2 \rangle - \langle F, \log(e^{x_1}e^{x_2}) \rangle + \langle F, x_1 \rangle    \\
 =& \langle F, e^{x_1} x_2 e^{-x_1} \rangle - \langle F, \log(e^{x_1}e^{x_2}) \rangle + \langle F, x_1 \rangle    \\
 =& \langle F, x_2 + [x_1, x_2] \rangle - \langle F, x_1 + x_2 +\tfrac{1}{2} [x_1, x_2] \rangle + \langle F, x_1 \rangle, \quad \text{up to degree 3}\\
 =& \tfrac{1}{2} \langle F, [x_1, x_2] \rangle  \\
 = & \tfrac{1}{2} \langle dA, [x_1, x_2] \rangle, \quad \text{up to degree 3}
\end{split}
\]
Hence, the two cohomology classes kill each other in the $\Delta$-complex, and the $k=2$ row is exact
\end{proof}

Now we are ready to prove Theorem \ref{thm:ascent}.

\begin{proof}
Let $\omega_1 \in \Omega^1\langle A, x_1, x_2 \rangle$ such that $d\Delta \omega_1 =0$. We now perform a zig-zag process in order to find the remaining terms in $\omega$. 
First, since the columns of the double complex are exact with the respect to the de Rham differential $d$, there is a unique element $\omega_0 \in \Omega^0\langle A, x_1, x_2, x_3 \rangle$ such that $d \omega_0 = - \Delta \omega_1$. This implies $d \Delta \omega_0 = - \Delta d\omega_0 = \Delta^2 \omega_1 =0$, and by exactness of the columns $\Delta \omega_0 =0$.

Next, note that $\Delta d \omega_1 = - d \Delta \omega_1=0$. Hence, by Lemma \ref{lem:delta-exact} there is an element $\omega_2 \in \Omega^2\langle A, x_1 \rangle$ such that $\Delta \omega_2 = - d \omega_1$. The element $\omega_2$ is unique by exactness of the 2nd row in the term $\Omega^2\langle A, x_1\rangle$.
Finally, since $\Delta d \omega_2= - d \Delta \omega_2= d^2 \omega_1=0$, by Lemma \ref{lem:delta-exact} there is an element $\omega_3 \in \Omega^3 \langle A \rangle$ such that $ \Delta \omega_3 = - d\omega_2$. Since the only $\Delta$-closed element in $\Omega\langle A\rangle$ is $\langle F, F\rangle$ (which is in degree 4), the choice of $\omega_3$ is unique.


Observe that  $\Delta d \omega_3 = - d \Delta \omega_3=0$. Hence, $d \omega_3 \in \Span \{ \langle F, F \rangle \}$ and $d\omega_3=\lambda \langle F, F \rangle$ for some  $\lambda \in \Bbbk$. In summary, for  $\omega=\omega_3 + \omega_2+ \omega_1 +\omega_0$ we have $D\omega = \lambda \langle F, F \rangle$. 
On the other hand, since the double complex $\Omega\langle A, x_\bullet \rangle$ is exact under the total differential $D$, the equation $D\omega = \langle F, F\rangle$ admits solutions and there is an element $\omega_1$ which yields $\lambda=1$.

\end{proof}

\subsection{From universal calculus to finite dimensional Lie algebras}

Let $\g$ be a finite dimensional Lie algebra over the field $\Bbbk$. 
Recall that the Weil algebra of $\g$ is a differential graded commutative algebra (together with a $T[1]G$-action) defined on
$$
W\g = S\g^* \otimes \wedge \g^*.
$$
Note that the space $W\g \otimes \g$ is naturally a differential graded Lie algebra, being a product of a Lie algebra with a ring. Define an element $a \in W\g \otimes \g$ as the canonical element (with respect to the bilinear form) in $\wedge^1\g^* \otimes \g \subset W\g \otimes \g$.

\begin{lem}
Let $\g$ be a finite dimensional Lie algebra.
Then, the assignment $A \mapsto a$ defines a homomorphism of differential graded Lie algebras $\mathcal{P}_\g: {\rm Lie}\langle A, dA\rangle \to W\g \otimes \g$. 
If $\g$ carries an invariant symmetric bilinear form $\langle \cdot, \cdot \rangle$, this homomorphism induces a chain map $\Omega\langle A\rangle \to W\g$.
\end{lem}

\begin{proof}
The first statement follows from the fact that ${\rm Lie}\langle A, dA\rangle$ is a free dg Lie algebra in one generator of degree 1. The second statement follows directly from the definition of $\mathcal{F}$.
\end{proof}

In a similar fashion, consider the cosimplicial complex $W\g \otimes \Omega^\bullet(G_{\rm formal})$, where $G_{\rm formal}$ is the formal group integrating the Lie algebra $\g$. Recall that in the formal group the logarithm $\log: G_{\rm formal} \to \g$ is well-defined.
We have the following  simple result:

\begin{lem}
The assignment $A \mapsto a, x_i \mapsto \log(g_i)$ defines a morphism of co-simplicial complexes
$$
\mathcal{P}_\g: (\Omega\langle A, x_\bullet\rangle, D= d+\Delta) \to (W\g \otimes \Omega(G_{\rm formal}^\bullet), D=d + \Delta).
$$
\end{lem}

\begin{proof}
The definitions of $d$ and $\Delta$ on the two sides match.
\end{proof}

This observation implies that calculations in the universal complex $(\Omega\langle A, x_\bullet\rangle, D= d+\Delta)$ automatically carry over to all co-simplicial complexes $(W\g \otimes \Omega(G_{\rm formal}^\bullet), D=d + \Delta)$. In particular, solving descent equations in the universal framework gives rise to solutions for all finite dimensional Lie algebras with an invariant symmetric  bilinear form.

%% file: sect2.tex
\section{The Kashiwara-Vergne theory}

\subsection{Derivations of free Lie algebras and the pentagon equation}

Let ${\rm der}_n$ be the Lie algebra of continuous derivations of $\Lie_n$ and ${\rm Aut}_n$ the group of continuous automorphisms of ${\rm Lie}_n$. Let $\tder_n$ be the vector space $\Lie_n^{\times n}$ equipped with the map
\begin{align*}
    \tder_n &\overset{\rho}{\to} \der_n \\
    (u_1, \cdots, u_n) &\mapsto u:(x_i \mapsto [u_i, x_i]).
\end{align*}
It is easy to see that the formula
\begin{equation} \label{eq:bracket_tder_n}
[(u_1, \cdots, u_n), (v_1, \cdots, v_n)]_i := \rho(u) v_i - \rho(v) u_i - [u_i, v_i].
\end{equation}
defines a Lie bracket on $\tder_n$ and makes $\rho$ into a Lie algebra homomorphism. This homomorphism defines an action of $\tder_n$ on ${\rm Lie}_n$, $\mathcal{F}\langle x_1, \dots, x_n\rangle$, $\Omega\langle x_1, \dots, x_n\rangle$ and other spaces where $\der_n$ acts. In particular, the action on $\Omega\langle x_1, \dots, x_n\rangle$ commutes with the de Rham differential:
$$
\rho(u) d\alpha = d(\rho(u) \alpha).
$$
We will often use a notation $u.\alpha = \rho(u)\alpha$ for actions of $\tder_n$ on various spaces.

Equipped with the  Lie bracket \eqref{eq:bracket_tder_n}, $\tder_n$ is a pro-nilpotent Lie algebra which readily integrates to a group denoted $\TAut_n$ together with the group homomorphism (again denoted by $\rho$):
\begin{align*}
    \rho: &\TAut_n \to \{g\in \Aut_n \ | \ g(x_i)=e^{\alpha_i}x_ie^{-\alpha_i}, \text{ for some } \ \alpha_i\in {\rm Lie}_n\}.
\end{align*}
In what follows we will need a Lie subalgebra of special derivations
$$
\sder_n=\{ u \in \tder_n \ | \ u.(x_1 + \dots + x_n)=0\} .
$$
This Lie algebra integrates to a group
$$
\SAut_n=\{ g \in \TAut_n \ | \ g.(x_1+ \dots x_n) = x_1 + \dots x_n\} .
$$

Define a vector space isomorphism $\gamma: \tder_n \to \Omega^1\left<x_1, \dots, x_n \right>$
$$
\gamma: u=(u_1, \dots, u_n) \mapsto \sum_i  \langle u_i, dx_i \rangle.
$$

The following Lemma is due to Drinfeld (\cite{drinfeld}):
\begin{lem}\label{lem:Drinfeld}
An element $u \in \tder_n$ is in $\sder_n$ if and only if $d \gamma(u)=0$.
\end{lem}


The construction described above has the following naturality property. Every partially defined map $f: \{1,\dots,n\} \to \{1,\dots, m\}$ induces a homomorphism of free Lie algebras ${\rm Lie}_m \mapsto {\rm Lie}_n$ defined by
$$
f^*: x_i \mapsto \sum_{f(l)=i} x_l.
$$
Furthermore, this map induces a map on differential forms
$$
f^*: \Omega\langle x_1, \dots, x_m\rangle \to \Omega\langle x_1, \dots, x_n \rangle.
$$
Composed with the inverse of $\gamma$, this map defines a Lie homomorphism  $f^*: \tder_m \to \tder_n$;
$$
    \left(u_i\right)_{i \in \{1, \dots, m\}} \longmapsto \left(u_{f(k)}(x_i = \sum_{f(l) = i} x_l) \right)_{k \in \{1,\dots,n\}}.
$$
The standard notation $f^*(u) = u^{f^{-1}(1), f^{-1}(2), \dots, f^{-1}(n)}$ can be used. For instance, 
$$
u^{12,3} = (u_1(x_1+x_2,x_3),u_1(x_1+x_2,x_3), u_2(x_1+x_2,x_3)).
$$

These maps easily integrate to maps between $\TAut_n$. For example, for  $g \in \TAut_2$ one can define an element $\Phi_g \in \TAut_3$ 
\begin{align}\label{eq:associator}
    \Phi_g &= (g^{12,3})^{-1} (g^{1,2})^{-1} g^{2,3} g^{1,23}.
\end{align}
Then part of Proposition 7.1 in \cite{AT} reads:
\begin{lem}
Let $g \in TAut_2$ and assume that
$$
g(x_1+x_2) = \log(e^{x_1}e^{x_2}).
$$
Then, $\Phi_g \in \SAut_3$ and it satisfies the pentagon equation
\begin{align}
    \Phi_g^{12,3,4} \Phi_g^{1,2,34} = \Phi_g^{1,2,3} \Phi_g^{1,23,4} \Phi_g^{2,3,4} \tag{\pentagon}.
\end{align}
\end{lem}

\subsection{The 1-cocycle  $c: \tder_n \to \Omega^1_n$}


In this section we define and study a 1-cocycle $c: \tder_n \to \Omega\langle x_1, \dots, x_n\rangle$ which plays a key role in constructing solutions of descent equations.

\begin{thm}
The map $c: \tder_n \to  \Omega^1\langle x_1, \dots, x_n\rangle$ defined by formula
\begin{equation} \label{eq:cocycle_c}
c: u=(u_1, \dots, u_n) \mapsto \sum_{i=1}^n \langle x_i,du_i\rangle
\end{equation}
is a 1-cocycle. That is, 
\begin{align*}
    c([u,v]) = u.c(v) - v.c(u) \quad \forall u,v \in \tder_n.
\end{align*}
Moreover, $c$ is natural in the sense that $c(f^*(u)) = f^*(c(u))$ for any partially defined map $f: \{1, \dots, m \} \to \{1, \dots, n\}$ of finite sets.
\end{thm}

\begin{proof}
First, we have $u.c(v)=\sum_{i=1}^n u.\langle x_i,dv_i\rangle=\sum (\langle [u_i,x_i],dv_i\rangle + \langle x_i,d(u.v_i)\rangle)$. Similarly, $v.c(u)=\sum v.\langle x_i,du_i\rangle=\sum_{i=1}^n (\langle [v_i,x_i],du_i\rangle + \langle x_i,d(v.u_i)\rangle)$.
Therefore, 
\begin{eqnarray*}u.c(v)-v.c(u)&=&\sum_{i=1}^n(\langle x_i,d(u.v_i)-d(v.u_i)\rangle+\langle [u_i,x_i],dv_i\rangle-\langle [v_i,x_i],du_i\rangle)\\&=&\sum_{i=1}^n(\langle x_i,d(u.v_i)-d(v.u_i)+d([v_i,u_i])\rangle \\
&=& c([u,v]),
\end{eqnarray*}
here we use the fact $\langle [u_i,x_i],dv_i\rangle-\langle [v_i,x_i],du_i\rangle=\langle x_i,[dv_i,u_i]\rangle-\langle x_i,[du_i,v_i]\rangle =\langle x_i, d[v_i, u_i]\rangle$.

For the naturality statement, let $f: \{1, \dots, m\} \to \{1, \dots, n\}$ be a partially defined map. On the one hand, we have
$$f^*(c(u))=\sum_{i=1}^nf^*\langle x_i,du_i\rangle=\sum_{i=1}^n\left(\sum_{k\in f^{-1}(i)}\langle x_k,du_i(x_j= \sum_{f(l) = j} x_l)\rangle\right).$$
Note that the summation is actually over $k$, and $i=f(k)$.

On the other hand, by the definition of $f^*:\tder_m{\longrightarrow} \tder_n$, we have
$$c(f^*(u))=\sum_{k=1}^n\langle x_k,du_{f(k)}(x_j= \sum_{f(l) = j} x_l)\rangle,$$
which implies the identity $c(f^*(u)) = f^*(c(u))$.
\end{proof}

\begin{rmk}
The map $c$ defines a nontrivial class in the Lie algebra cohomology 
$$
[c] \in H^1(\tder_n, \Omega^1\langle x_1, \dots, x_n\rangle).
$$
Indeed, 
the degree of $c$ is equal to $+1$. If $c$ were a trivial 1-cocycle, it would have been of the form $c(u) = u.\alpha$ for some $\alpha \in \Omega^1\langle x_1, \dots, x_n\rangle$ of degree $+1$. However, such elements do not exist.
\end{rmk}


\begin{lem}
The cocycle $c$ is a bijection. Furthermore, it restricts to a bijection 
$$c : \sder_n \to \Omega^{1,\text{closed}}\langle x_1, \dots, x_n\rangle \cong \Omega^0\langle x_1, \dots, x_n\rangle.$$

\end{lem}
\begin{proof}
Consider the map $c \circ \gamma^{-1}: \Omega^1_n \to \Omega^1_n$. It is equal to $c\circ \gamma = d \circ e - {\rm Id}$ (where $e$ is the derivation of $\Omega_n$ defined by $e(dx_i)=x_i, e(x_i)=0$). Indeed,
$$
(d\circ e -{\rm Id}) \sum_{i=1}^n \langle u_i, dx_i \rangle = d \sum_{i=1}^n \langle u_i, x_i \rangle - \sum_{i=1}^n \langle u_i, dx_i \rangle = \sum_{i=1}^n
\langle x_i, du_i\rangle .
$$

Note that
$$
(d \circ e - {\rm Id})({\rm Id} - e \circ d) = n -{\rm Id},
$$
where $n = de+ed$ is the Euler derivation counting the number of generators. Since, the number of generators is always greater or equal to two, $n -{\rm Id}$ is an invertible operator. Hence, $c$ is a bijection.

Note that 
$$
\gamma(u) + c(u) = d\sum_{i=1}^n \langle x_i, u_i\rangle.
$$
Hence, $c(u)$ is closed if and only if $\gamma(u)$ is closed, and $\gamma(u)$ is closed if and only if $u \in \sder_n$ as follows from Drinfeld's Lemma \ref{lem:Drinfeld}.
\end{proof}

The Lie algebra $1$-cocycle $c$ integrates to a Lie group cocycle, which is the map $C:\TAut_n \to \Omega^1_n$ uniquely determined by the following properties:
\begin{align*}
    C(g \circ f) &= C(g) + g. C(f), \\
    \left. \frac{d}{dt}\right\vert_{t=0} C(e^{t u}) &= c(u).
\end{align*}
The map $C$ is given by the explicit formula,
\begin{align*}
C(e^u)=\left(\frac{e^u-1}{u}\right).c(u), \ \forall u\in \tder_n.
\end{align*}


\begin{pro}
The Lie group cocycle $C : \TAut_n \to \Omega^1$ is a bijection. It restricts to a bijection $C : \SAut_n \to \Omega^{1,\text{closed}}_n \cong \Omega^0_n$.
\end{pro}
\begin{proof}
This follows directly from the fact that the action of $\tder_n$ on $\Omega^1_n$ is of positive degree and from the bijectivity of $c$.\end{proof}

\subsection{Solving descent equations with Kashiwara-Vergne theory}

Let us define $S=\{g\in \TAut_2~|~g(x_1+x_2)=\log (e^{x_1}e^{x_2})\}$, the set of solutions to the first equation in the Kashiwara-Vergne problem. 

\begin{lem}    \label{lem:g_in_S}
Any $g\in S$ satisfies $d\Delta C(g) = 0$.
\end{lem}

\begin{proof}
Since $g \in S$ is a solution of the first Kashiwara-Vergne equation, its associator \eqref{eq:associator} satisfies $\Phi_g \in \SAut_3$. Set $\omega = C(g)$ and apply the cocycle $C$ to the equation $g^{2,3} g^{1,23} = g^{1,2} g^{12,3} \Phi_g$ to get
\begin{align*}
0 &= C(g^{2,3} g^{1,23}) - C(g^{1,2} g^{12,3}) \\
&= (\omega^{2,3} + g^{2,3}. \omega^{1,23}) - (\omega^{1,2} + g^{1,2}. \omega^{12,3} + (g^{1,2} g^{12,3}). C(\Phi_g)) \\
&=  (\omega(x_2, x_3) + g^{2,3}. \omega(x_1, x_2 + x_3)) - (\omega(x_1, x_2) + g^{1,2}. \omega(x_1 + x_2, x_3) + (g^{1,2} g^{12,3}). C(\Phi_g)) \\
&= (\omega(x_2, x_3) + \omega( x_1, \log(e^{x_2}e^{x_3}))) - (\omega(x_1, x_2) + \omega( \log(e^{x_1}e^{x_2}), x_3) + (g^{1,2}g^{12,3}).C(\Phi_g)) \\
 &=   \Delta(\omega) - (g^{1,2}g^{12,3}).C(\Phi_g).
\end{align*}
Note that the universal 1-form  $(g^{1,2}g^{12,3}).C(\Phi_g)$ is closed  since $C$ maps $\Phi_g \in \SAut_3$ to $\Omega^{1, {\rm closed}}_3$, and $\Omega^{1, {\rm closed}}$  is preserved by every automorphism (in particular, by $g^{1,2}g^{12,3}$).
Hence, 
\begin{equation} \label{eq:ggdC}
d \Delta \omega = d \big( (g^{1,2}g^{12,3}).C(\Phi_g) \big) = (g^{1,2}g^{12,3}). dC(\Phi_g) =0.
\end{equation}
\end{proof}

For $g \in S$, Lemma \ref{lem:g_in_S} implies that the element $\omega_1 =C(g) \in \Omega^1\langle x_1, x_2\rangle$ verifies the conditions of Theorem \ref{thm:ascent}.
Hence, there is a unique $\omega=\omega_0+\omega_1 + \omega_2 + \omega_3$ such that its components are solutions of the descent equations.
In particular, $\omega_0 \in \Omega^0\langle x_1, x_2, x_3 \rangle$ is a $\Delta$-cocycle.
By Lemma \ref{lem:cohgen} the cohomology $H^3(\Omega^0, \Delta)$ is spanned by the generator $\phi= \langle x_1, [x_2, x_3]\rangle + \dots$. The following lemma shows that the cohomology class $[\omega_0] \in H^3(\Omega^0, \Delta)$ is independent of the choice of $g \in S$:
\begin{lem}
The cohomology class $[\omega_0] \in H^3(\Omega^0, \Delta)$ is independent of $g$. More precisely,
\begin{align*}
    [\omega_0]=\frac{1}{12} [\phi].
\end{align*}
\end{lem}
\begin{proof}
Recall that $S$ is a right torsor under the action of the group $\SAut_2$. If one chooses  a base point $g \in S$, all other solutions of the first Kashiwara-Vergne equation are of the form $g'=g\circ f$ for $f \in \SAut_2$.  The cocycle condition for $C$  yields
$$
C(g \circ f) - C(g) = g.C(f).
$$
The 1-form $C(f)$ is closed (since $f \in \SAut_2$) and therefore exact. Hence, $g.C(f)$ is also exact, and we will denote it by $d \nu$. Applying the differential $\Delta$ to the equation above, we obtain
$$
\Delta C(g \circ f) - \Delta C(g) = \Delta d \nu = d \Delta \nu.
$$
This implies
$$
d\omega_0(g') - d\omega_0(g)= \Delta C(g') - \Delta C(g) = d \Delta \nu.
$$
Since the kernel of $d$ in degree 0 is trivial, we obtain 
$$
\omega_0(g') - \omega_0(g) = \Delta \nu,
$$
as required.

In order to compute the missing coefficient, we use equation \eqref{eq:ggdC} to conclude
$$
d\omega_0 =  \Delta \omega_1 =  (g^{1,2} g^{12,3}).C(\Phi_g).
$$
If $\Phi_g$ is a Drinfeld associator ({\em e.g.} the Knizhnik-Zamolodchikov associator), we have
$\Phi_g=\exp(u_2 + \dots)$, where $u_2 \in \tder_3$ is the following tangential derivation of degree 2:
$$
u_2 = \frac{1}{24} \, ([x_2, x_3], [x_3, x_1], [x_2, x_3])
$$
and $\dots$ stand for higher degree terms (see Propositions 7.3 and 7.4 in \cite{AT}).
Up to degree 2, the equation for $\omega_0$ reads
$$
\begin{array}{lll}
d\omega_0 & = & c(u_2) + \dots \\
 &= & 
 \frac{1}{24} \, \left( \langle d [x_2, x_3], x_1\rangle + \langle d [x_3, x_1], x_2\rangle +
\langle d [x_1, x_2], x_3 \rangle \right) + \dots \\
& =&   \frac{1}{12} \, d \langle x_1, [x_2, x_3] \rangle + \cdots
\end{array}
$$
which implies
$$
\omega_0 = \frac{1}{12} \, \langle x_1, [x_2, x_3] \rangle + \cdots ,
$$
as required.
\end{proof}

\begin{pro}  \label{prop:gomega}
The cocycle $C$ defines a bijection
\begin{align*}
    S & = \{g\in \TAut_2 \ |\ g(x_1+x_2)=\log (e^{x_1}e^{x_2})\}  \\
    & \overset{C}{\longrightarrow} \{ \omega^1 \in \Omega^1\langle x_1, x_2 \rangle \ | \ \Delta(\omega^1) = d\omega_0, [\omega_0] = \tfrac{1}{12}[\phi] \in H^3(\mathcal{F}\langle x_\bullet \rangle, \Delta) \}.
\end{align*}
\end{pro}

\begin{proof}
The cocycle $C$ is an injective map. Hence, so is its restriction to $S$. It remains to show the surjectivity.

Let $\omega_1, \tilde{\omega}_1 \in \Omega^1\langle x_1, x_2\rangle$ be two elements which satisfy conditions of the Proposition, and assume that $\omega_1 = C(g)$ for some $g \in S$. The corresponding elements $\omega_0, \tilde{\omega}_0 \in \Omega^0\langle x_1, x_2\rangle$ belong to the same cohomology class in $H^3(\Omega^0, \Delta)$, and therefore $\tilde{\omega}_0 - \omega_0 = \Delta \mu$ for some $\mu \in \Omega^0\langle x_1, x_2\rangle$. This implies
$$
\Delta(\tilde{\omega}_1 - \omega_1)=d(\tilde{\omega}_0 - \omega_0) = d \Delta \mu = \Delta d\mu.
$$ 

We know that the kernel of  $\Delta: \Omega^1\langle A, x_1, x_2\rangle \to \Omega^1\langle A,  x_1, x_2, x_3\rangle$ is spanned by the elements $ \Delta d \langle x_1, x_1\rangle, \Delta \langle A, x_1\rangle$. The first of them is $d$-exact, and the second one contains $A$. Therefore, the kernel of the map $\Delta: \Omega^1\langle  x_1, x_2\rangle \to \Omega^1\langle  x_1, x_2, x_3\rangle$ is spanned by the $d$-exact element
$d \Delta \langle x_1, x_1\rangle$.

We conclude that $\tilde{\omega}_1 - \omega_1 = d \sigma$ is $d$-exact. Choosing $f \in \SAut_2$ such that $C(f) = d(g^{-1}.\sigma)=g^{-1}.d\sigma$, we obtain
$$
C(g \circ f) = C(g) + g.C(f)=\omega_1 + d\sigma = \tilde{\omega}_1,
$$
as required.
\end{proof}

\begin{rmk}
The proof furthermore shows that elements of $S$ are up to the action of $e^{\Bbbk d\langle x_1, x_1 \rangle}$ in bijection with the set of representatives of the cohomology class $\tfrac{1}{12}[\phi]$. Moreover, one can use results from \cite{AT} to conclude that there is a bijection between normalized associators in $\SAut$ and normalized solutions to $\Delta \omega_0 = 0$. 
\end{rmk}

Together with theorem \ref{thm:ascent} one concludes
\begin{thm}
Primitives of $\tfrac{1}{2} \langle F, F \rangle \in \Omega^\bullet\langle A, x_\bullet \rangle$ with respect to $d + \Delta$ are of the form $\omega_0 +  \omega_1 + \omega_2 + \omega_3$, where
\begin{align*}
    \omega_3 &= \tfrac{1}{2} CS(A)\\
    \omega_2 &= \tfrac{1}{2} WZ(A,e^{x_1}) + s d \langle A,x_1 \rangle\\
    \omega_1 &= -C(g) - s \Delta \langle A,x_1 \rangle
\end{align*}
for uniquely determined $g\in S$ and $s \in \Bbbk$. Moreover, $[\omega_0] = -\tfrac{1}{12} [\phi]$.
\end{thm}

\begin{proof}
By Proposition \ref{prop:gomega}, every element $g\in S$ gives rise to $\omega_1 = C(g)$ which satisfies conditions of Theorem \ref{thm:ascent}. Hence, it defines a unique element $\omega=\omega_0+\omega_1+\omega_2+\omega_3$ with the property $D\omega = \lambda \langle F, F \rangle$. 
Consider $g=\exp(u_1 + \dots)$, where $u_1=\frac{1}{2}(x_2, 0)$ is the tangential derivation of degree one, and $\dots$ stand for higher order terms. It is sufficient to keep the terms up to degree 2 which yields
$$
\omega_0=0 + \dots, \omega_1 = -\frac{1}{2}\, \langle x_1, dx_2 \rangle + \dots, \omega_2 = \frac{1}{2} \langle A, dx_1 \rangle + \dots, \omega_3 = \frac{1}{2} \langle A, dA\rangle + \cdots
$$
which shows that $\lambda = 1/2$, as required.

The  primitives of $\frac{1}{2} \langle F, F\rangle$ form an affine space over space of $D$-exact elements of total degree 3. This space is of the form
$$
D( \Bbbk d\langle x_1, x_1 \rangle \oplus \langle A, x_1\rangle) + D \Omega^0\langle x_1, x_2\rangle .
$$
The sum in not direct since $d\Delta \langle x_1, x_1\rangle$ spans the intersection of the two subspaces. We have already classified all solutions with $\omega_1 \in \Omega^1\langle x_1, x_2\rangle$. Hence, a general solution for $\omega_1 \in \Omega^1\langle A, x_1, x_2\rangle$ which verifies $d\Delta \omega_1 =0$ is of the form
$$
\omega_1 = -C(g) - s \Delta \langle A, x_1 \rangle .
$$
Formulas for $\omega_2$ and $\omega_3$ follow from the descent equations.
\end{proof}

\begin{rmk}
Let $g \in S$. Then, there is the following interesting identity relating the pentagon equation for $\Phi_g$ to the last descent equation $\Delta \omega_0=0$ for $\omega_0$ defined by $g$:
$$
g^{1,2}g^{12,3}g^{123,4}.\left( C(\Phi_g^{12,3,4}\Phi_g^{1,2,34}) - C(\Phi_g^{1,2,3} \Phi_g^{1,23,4}\Phi_g^{2,3,4})\right) = \Delta d\omega_0 = d \Delta \omega_0.
$$

\end{rmk}
